\documentclass[a4paper]{siamart1116}

\usepackage{dsfont}
\usepackage{amsmath}
\usepackage{amssymb}

\usepackage{amsfonts}
\usepackage{graphicx}
\usepackage{epstopdf}
\usepackage{algorithmic}
\Crefname{ALC@unique}{Line}{Lines} 

\ifpdf
  \DeclareGraphicsExtensions{.eps,.pdf,.png,.jpg}
\else
  \DeclareGraphicsExtensions{.eps}
\fi

\usepackage{geometry}
\numberwithin{theorem}{section}

\newcommand{\TheTitle}{A note on computing range space bases of rational matrices}
\newcommand{\TheAuthors}{A. Varga}

\headers{\TheTitle}{\TheAuthors}

\title{{\TheTitle}}

\author{
  Andreas Varga\thanks{Gilching, Germany
    (\email{varga.andreas@gmail.com}).}
}
\headers{Range computation of rational matrices}{A. Varga}

\ifpdf
\hypersetup{
  pdftitle={\TheTitle},
  pdfauthor={\TheAuthors}
}
\fi

\newcommand{\be}{\begin{equation}}
\newcommand{\ee}{\end{equation}}
\newcommand{\ba}{\left [ \begin{array}}
\newcommand{\ea}{\end{array} \right ]}
\newcommand{\bea}{\begin{eqnarray}}
\newcommand{\eea}{\end{eqnarray}}

\newcommand{\rank}{\mathop{\mathrm{rank}}}
\newcommand{\diag}{\mathop{\mathrm{diag}}}

\def\iu{\ensuremath{\mathrm{i}}}

\begin{document}

\maketitle

\begin{abstract}
We discuss computational procedures based on descriptor state-space realizations to compute proper range space bases of rational matrices. The main computation is the orthogonal reduction of the system matrix pencil to a special Kronecker-like form, which allows to extract a full column rank factor, whose columns form a proper rational basis of the range space. The computation of several types of bases can be easily accommodated, such as minimum-degree bases, stable inner minimum-degree bases, etc. Several straightforward applications of the range space basis computation are discussed, such as, the computation of full rank factorizations,  normalized coprime factorizations, pseudo-inverses, and inner-outer factorizations.
\end{abstract}



\section{Introduction}
For any $p\times m$ real rational matrix $G(\lambda)$ of normal rank $r$, there exists a \emph{full-rank factorization} of $G(\lambda)$ of the form
\be\label{full-rank-fac} G(\lambda) = R(\lambda)X(\lambda) , \ee
where $R(\lambda)$ is a $p\times r$ full column rank rational matrix and $X(\lambda)$ is a $r\times m$ full row rank rational matrix. This factorization generalizes the full-rank factorization of constant matrices, and, similarly to the constant case, it is not unique. Indeed, for any $r\times r$ invertible rational matrix $M(\lambda)$, $G(\lambda) = \widetilde R(\lambda)\widetilde X(\lambda)$, with $\widetilde R(\lambda) = R(\lambda)M^{-1}(\lambda)$ and $\widetilde X = M(\lambda)X(\lambda)$, is also a full-rank factorization of $G(\lambda)$.

The existence of the full-rank factorization (\ref{full-rank-fac}) can be inferred from the Smith-McMillan form of $G(\lambda)$ \cite{Kail80}, which also indicates that both the poles as well as the zeros of $R(\lambda)$ can be arbitrarily chosen. In particular, the zeros of $G(\lambda)$ can be split between the two factors in (\ref{full-rank-fac}), such that $R(\lambda)$ only includes a selected set of zeros, while $X(\lambda)$ includes the rest of zeros. A minimum-degree $R(\lambda)$ corresponds to the complete absence of zeros in $R(\lambda)$.

Using (\ref{full-rank-fac}), it is straightforward to show  that $G(\lambda)$ and $R(\lambda)$ have the same range space over the rational functions, i.e.
\[ \mathcal{R}(G(\lambda)) = \mathcal{R}(R(\lambda)).  \]
For this reason, with a little abuse of language, we will call $R(\lambda)$ the range (or image) matrix of $G(\lambda)$ (or simply the range of $G(\lambda)$). It follows, that for each rational column vector $y(\lambda) \in \mathcal{R}(G(\lambda))$, there exists  $x(\lambda) \in \mathcal{R}(R(\lambda))$ such that $R(\lambda) x(\lambda) = y(\lambda)$. Since $R(\lambda)$ has full column rank $r$, its columns form a set of $r$ basis vectors of $\mathcal{R}(G(\lambda))$.

In this note, we describe a general computational approach based on a descriptor state-space realization of the rational matrix $G(\lambda)$ to determine a full column rank $R(\lambda)$, representing a proper range space basis of $G(\lambda)$. The zeros of $R(\lambda)$ can be enforced to lie in a specified domain of the complex plane $\mathds{C}_b$. The main computation is the orthogonal reduction of the corresponding system matrix pencil to a special Kronecker-like form (already employed in \cite{Oara00} and \cite{Oara05}), which allows to immediately extract a full column rank factor $R(\lambda)$, which includes all zeros of $G(\lambda)$ lying in $\mathds{C}_b$. Straightforward applications of the range computation techniques are mentioned and numerical examples are given.

\section{Range computation}
Let $G(\lambda)$ be a $p\times m$ real rational matrix. We can associate $G(\lambda)$  with the \emph{transfer function matrix} (TFM) of a generalized linear time-invariant system (or descriptor system), where, for a continuous-time system, the frequency variable has the significance $\lambda = s$, the complex variable in the Laplace-transform, and, for a discrete-time system, the frequency variable has the significance $\lambda = z$, the complex variable in the $\mathcal{Z}$-transform.
The underlying descriptor system has a generalized state-space representation of the form
\be\label{app:dss}
\begin{array}{rcl} E \lambda x(t) &=& Ax(t) + Bu(t) , \\
y(t) &=& Cx(t) + Du(t) ,
\end{array} \ee
where $x(t) \in \mathds{R}^n$ is the state vector, $u(t) \in \mathds{R}^m$ is the input vector, and $y(t) \in \mathds{R}^p$ is the output vector, and where $\lambda$ is either the differential operator $\lambda x(t) = \frac{\text{d}}{\text{d}t}x(t)$ for a continuous-time system or the advance operator
$\lambda x(t) = x(t+1)$ for a discrete-time system. In all what follows, we assume
$E$ is square and possibly singular, and the pencil $A-\lambda E$ is regular (i.e., $\det (A-\lambda E) \not \equiv 0$). The descriptor system  (\ref{app:dss}) represents a state-space realization of the TFM $G(\lambda)$ if
\be\label{GTFM} G(\lambda) = C(\lambda E-A)^{-1}B+D .\ee
We will also use the equivalent notation for the TFM in (\ref{GTFM})
\be\label{GTFMalt} G(\lambda) = \ba{c|c} A-\lambda E & B\\ \hline C & D \ea .\ee
The descriptor system (\ref{app:dss}) will be alternatively denoted by the quadruple $(A-\lambda E,B,C,D)$.

We recall from \cite{Verg79,Verg81} some basic notions related to descriptor system realizations.  A  realization $(A-\lambda E,B,C,D)$ is \emph{minimal} if it is controllable, observable and has no non-dynamic modes. A controllable and observable realization is called \emph{irreducible}. The \emph{poles} of $G(\lambda)$ are related to $\Lambda(A-\lambda E)$, the eigenvalues of the pencil $A-\lambda E$ (also known as the generalized eigenvalues of the pair $(A,E)$).
For a minimal  realization, the finite poles of $G(\lambda)$ are the finite eigenvalues of $A-\lambda E$, while the multiplicities of the infinite poles of $G(\lambda)$ are defined by the multiplicities of the infinite eigenvalues of $A-\lambda E$ minus one. A finite eigenvalue $\lambda_f \in \Lambda(A-\lambda E)$ is controllable if $\rank\,[\,A-\lambda_f E \; B\,] = n$, otherwise is uncontrollable. Similarly, a finite eigenvalue $\lambda_f \in \Lambda(A-\lambda E)$ is observable if $\rank\,[\,A^T-\lambda_f E^T \; C^T\,] = n$, otherwise is unobservable. Infinite controllability requires that $\rank\,[\, E \; B\,] = n$, while infinite observability requires that $\rank\,[\,E^T \; C^T\,] = n$. The lack of non-dynamic modes can be equivalently expressed as $A \mathcal{N}(E) \subseteq \mathcal{R}(E)$, where $\mathcal{N}(E)$ denotes the right nullspace of $E$. The zeros of $G(\lambda)$ are related to the eigenvalues of the system matrix pencil
\be\label{sys-pencil} S(\lambda) = \ba{cc} A-\lambda E & B \\ C & D \ea . \ee
For a minimal  realization, the finite zeros of $G(\lambda)$ are the finite eigenvalues of $S(\lambda)$, while the multiplicities of the infinite zeros of $G(\lambda)$ are defined by the multiplicities of the infinite eigenvalues of $S(\lambda)$ minus one.

Consider a disjunct partition of the complex plane $\mathds{C}$ as
\be\label{Cgoodbad}  \mathds{C} = \mathds{C}_g \cup \mathds{C}_b, \quad \mathds{C}_g \cap \mathds{C}_b = \emptyset \, ,\ee
where $\mathds{C}_g$ and $\mathds{C}_b$ are symmetric with respect to the real axis.
$\mathds{C}_g$ and $\mathds{C}_b$ are usually associated with the ``good'' and ``bad'' domains of the complex plane $\mathds{C}$ for the poles and zeros of $G(\lambda)$.
We say the descriptor system (\ref{GTFMalt}) is \emph{proper} $\mathds{C}_g$-\emph{stable} if all  finite eigenvalues of $A-\lambda E$ belong to $\mathds{C}_g$ and all infinite eigenvalues of $A-\lambda E$ are simple. The descriptor system (\ref{GTFMalt}) (or equivalently the pair $(A-\lambda E,B)$) is $\mathds{C}_b$-\emph{stabilizable} if $\rank \ba{cc} A-\lambda E & B \ea = n$ for all finite $\lambda \in \mathds{C}_b$ and $\rank\,[\, E \; B\,] = n$. The descriptor system (\ref{GTFMalt})   (or equivalently the pair $(A-\lambda E,C)$)  is $\mathds{C}_b$-\emph{detectable} if $\rank \left[\begin{smallmatrix} A-\lambda E \\C \end{smallmatrix}\right] = n$ for all finite $\lambda \in \mathds{C}_b$ and $\rank\,[\,E^T \; C^T\,] = n$.

The following result slightly extends \cite[\textbf{Theorem 2.2}]{Oara05} and is instrumental for the suggested computational approach of proper range space bases.

\begin{lemma} \label{lem:sklf} Let $G(\lambda)$ be a $p\times m$ real rational matrix of normal rank $r$, with a $\mathds{C}_b$-{stabilizable} descriptor system realization  $(A-\lambda E,B,C,D)$  satisfying (\ref{GTFM}). Then, there exist two orthogonal matrices $U$ and $Z$ such that
\be\label{spec-klf}
\ba{cc} U & 0 \\ 0 & I \ea \ba{cc} A - \lambda E & B \\ \hline C & D \ea Z =
\ba{cccc} A_{rg}-\lambda E_{rg} & \ast & \ast & \ast \\
0 & A_{b\ell}-\lambda E_{b\ell} & B_{b\ell} & \ast \\
0 & 0 & 0 & B_n \\ \hline
0 & C_{b\ell} & D_{b\ell} & \ast \ea , \ee
where
\begin{itemize}
\item[(a)] The pencil $A_{rg}-\lambda E_{rg}$ has full row rank for $\lambda \in \mathds{C}_g$ and $E_{rg}$ has full row rank.
\item[(b)]  $E_{b\ell}$ and $B_n$ are invertible, the pencil
\be\label{syspencil} \ba{cc}   A_{b\ell}-\lambda E_{b\ell} & B_{b\ell} \\ C_{b\ell} & D_{b\ell} \ea \ee
has full column rank $n_{b\ell}+r$ in $\mathds{C}_g$ and the pair $(A_{b\ell}-\lambda E_{b\ell}, B_{b\ell})$ is $\mathds{C}_b$-stabilizable.
\end{itemize}
\end{lemma}

This lemma allows to construct the range of $G(\lambda)$ using the following result.
\begin{theorem}
Let $G(\lambda)$ be a $p\times m$ real rational matrix of normal rank $r$, with the $\mathds{C}_b$-{stabilizable} descriptor system realization $(A-\lambda E,B,C,D)$ satisfying (\ref{GTFM}). Let $U$ and $Z$ be the orthogonal matrices used in Lemma \ref{lem:sklf} to obtain the system matrix pencil in the special Kronecker-like form (\ref{spec-klf}). Then, the range matrix of $G(\lambda)$ which includes the zeros of $G(\lambda)$ in $\mathds{C}_b$ has the proper descriptor system realization
\be\label{range} R(\lambda) = \ba{c|c} A_{b\ell}-\lambda E_{b\ell} & B_{b\ell} \\ \hline C_{b\ell} & D_{b\ell} \ea . \ee
\end{theorem}
\begin{proof} Since, by construction, $R(\lambda)$ has full column rank and contains all zeros of $G(\lambda)$ in $\mathds{C}_b$, we have only to show that
there exists $X(\lambda)$ which satisfies the linear rational matrix equation (\ref{full-rank-fac}). This comes down to show that the compatibility condition
\be\label{rank-comp} \rank R(\lambda) = \rank [\, R(\lambda) \; G(\lambda) \,] = r \ee
is fulfilled.
A descriptor system realization of  $[\, R(\lambda) \; G(\lambda) \,]$ is
\[  [\, R(\lambda) \; G(\lambda) \,] = \ba{cc|cc} A_{b\ell}-\lambda E_{b\ell} & 0 & B_{b\ell} & 0 \\
0 & A-\lambda E & 0 & B \\ \hline
C_{b\ell} & C & D_{b\ell} & D \ea  \]
and the rank condition (\ref{rank-comp}) is equivalent to
\be\label{rank-cond} \rank \ba{cccc} A_{b\ell}-\lambda E_{b\ell} & B_{b\ell} & 0 &0 \\
0 & 0 & A-\lambda E & B \\
C_{b\ell} &  D_{b\ell} & C & D \ea = n_{b\ell}+n+r . \ee
By premultiplying the pencil
\[ S(\lambda) := \ba{cccc} A_{b\ell}-\lambda E_{b\ell} & B_{b\ell} & 0 &0 \\
0 & 0 & A-\lambda E & B \\
C_{b\ell} &  D_{b\ell} & C & D \ea \]
with $\widetilde U = \diag (I_{n_{b\ell}},U,I)$ and postmultiplying it with $\widetilde Z = \diag (I_{n_{b\ell}+r},Z)$ we obtain
\[ \widetilde S(\lambda) := \widetilde U S(\lambda)\widetilde Z = \ba{cccccc}
A_{b\ell}-\lambda E_{b\ell} & B_{b\ell} & 0 &0 & 0 & 0\\
0 & 0 & A_{rg}-\lambda E_{rg} & \ast & \ast & \ast \\
0 & 0 & 0 & A_{b\ell}-\lambda E_{b\ell} & B_{b\ell} & \ast \\
0 & 0 & 0 & 0 & 0 & B_n \\
C_{b\ell} &  D_{b\ell} & 0 & C_{b\ell} & D_{b\ell} & \ast \ea \]
To prove (\ref{rank-cond}), we show that $\rank \widetilde S(\lambda) = n_{b\ell}+n+r$, by performing successive block row and block column operations which preserve its rank. The first three block operations are:
\begin{itemize}
\item[1)] subtract the first block column multiplied from right with $(A_{b\ell}-\lambda E_{b\ell})^{-1}B_{b\ell}$ from the second  block column;
\item[2)] subtract the resulting first block row multiplied from left with $C_{b\ell}(A_{b\ell}-\lambda E_{b\ell})^{-1}$ from the last block row;
    \item[3)] subtract the third block row multiplied from left with $C_{b\ell}(A_{b\ell}-\lambda E_{b\ell})^{-1}$ from the last block row.
\end{itemize}
After performing these operations, we obtain
\[ \rank \widetilde S(\lambda) = \rank \ba{cccccc}
A_{b\ell}-\lambda E_{b\ell} & 0 & 0 &0 & 0 & 0\\
0 & 0 & A_{rg}-\lambda E_{rg} & \ast & \ast & \ast \\
0 & 0 & 0 & A_{b\ell}-\lambda E_{b\ell} & B_{b\ell} & \ast \\
0 & 0 & 0 & 0 & 0 & B_n \\
0 &  R(\lambda) & 0 & 0 & R(\lambda) & D_n(\lambda) \ea
 \]
where $R(\lambda) = C_{b\ell} (\lambda E_{b\ell}-A_{b\ell})^{-1}B_{b\ell} +D_{b\ell}$ and $D_n(\lambda)$ denotes the resulting rational matrix in the last block of the last block row. We continue the reduction of the resulted rational matrix by performing two additional operations:
\begin{itemize}
\item[4)]  subtract the fourth block row multiplied from left with $B_n^{-1}D_n(\lambda)$ from the last block row;
\item[5)]  subtract the second block column from the fifth block column.
\end{itemize}
We finally obtain
\[ \rank \widetilde S(\lambda) = \rank \ba{c|c|cccc}
A_{b\ell}-\lambda E_{b\ell} & 0 & 0 &0 & 0 & 0\\ \hline
0 & 0 & A_{rg}-\lambda E_{rg} & \ast & \ast & \ast \\
0 & 0 & 0 & A_{b\ell}-\lambda E_{b\ell} & B_{b\ell} & \ast \\
0 & 0 & 0 & 0 & 0 & B_n \\ \hline
0 &  R(\lambda) & 0 & 0 & 0 & 0 \ea ,
 \]
from which we immediately have the desired result by observing that
\[ \begin{array}{lll}\rank \widetilde S(\lambda)& = & \rank (A_{b\ell}-\lambda E_{b\ell})+ \rank R(\lambda) + \rank {\arraycolsep=.8mm\ba{cccc}A_{rg}-\lambda E_{rg} & \ast & \ast & \ast \\
0 & A_{b\ell}-\lambda E_{b\ell} & B_{b\ell} & \ast \\
0 & 0 & 0 & B_n \\ \ea} \\
&=& n_{b\ell} + r + n \, . \end{array}
\]
\end{proof}

For the computation of the descriptor realization (\ref{range}) of the range $R(\lambda)$, a numerically stable algorithm can be devised, which exclusively uses orthogonal transformations to reduce the system matrix pencil to the special form (\ref{spec-klf}). The main steps of such an algorithm are given in the (constructive) proof of \textbf{Theorem 2.2} in \cite{Oara05}. The basic ingredients of such an algorithm are: (a) column and row compressions to full column rank or full row rank matrices, respectively, performed via QR-factorizations with column pivoting, or, more reliably, using singular value decompositions; (b) reduction of a linear pencil to a Kronecker-like staircase form using orthogonal similarity transformations, such that the right, regular and left Kronecker structures are separated; (c) reordering of the eigenvalues of the regular part using orthogonal similarity transformations via the QZ-algorithm. Suitable computational algorithms are described in \cite{Golu13} for (a) and (c), and in \cite{Oara97} for (b) (see also \cite[Chapter 10]{Varg17} for an overview of these techniques).

With an additional similarity transformation of the form
\[ Q = \diag \left(I, \ba{cc} I_{n_{b\ell}} & 0 \\ F & I_r \ea, I\right) \]
we achieve
\be\label{spec-klf2}
\ba{cc} U & 0 \\ 0 & I \ea \ba{cc} A - \lambda E & B \\ \hline C & D \ea ZQ =
\ba{cccc} A_{rg}-\lambda E_{rg} & \ast & \ast & \ast \\
0 & A_{b\ell}+B_{b\ell}F-\lambda E_{b\ell} & B_{b\ell} & \ast \\
0 & 0 & 0 & B_n \\ \hline
0 & C_{b\ell}+D_{b\ell}F & D_{b\ell} & \ast \ea . \ee
It follows that, with an arbitrary invertible $W$,
\be\label{range_stab} R(\lambda) = \ba{c|c} A_{b\ell}+B_{b\ell}F-\lambda E_{b\ell} & B_{b\ell}W \\ \hline C_{b\ell}+D_{b\ell}F & D_{b\ell}W \ea  \ee
is also a proper range of $G(\lambda)$. Since the pair $(A_{b\ell}-\lambda E_{b\ell}, B_{b\ell})$ is $\mathds{C}_b$-stabilizable, we can choose $F$ such that all eigenvalues of $A_{b\ell}-\lambda E_{b\ell}$ lying in $\mathds{C}_b$, can be moved to arbitrary locations in an appropriate stability domain of $\mathds{C}$.

\section{Some applications}
\subsection{Full rank factorizations} A full rank factorization of an arbitrary rational matrix $G(\lambda)$ of normal rank $r$ of the form (\ref{full-rank-fac}) can be determined with a proper $R(\lambda)$ of the form (\ref{range_stab}) and $X(\lambda)$ of the form
\be\label{full-row-rank} X(\lambda) = \ba{c|c} A-\lambda E & B \\ \hline \\[-3.5mm] \widetilde C & \widetilde D \ea , \ee
where $[\,\widetilde C \; \widetilde D\, ] = W^{-1}[\, 0 \;\; -F \;\; I_r \;\; 0 \,]Z^T$. The expression of $X(\lambda)$ can be verified by explicitly computing the descriptor realization of $R(\lambda)X(\lambda)$ (see the proof of \textbf{Theorem 3.1} in \cite{Oara05}).

A dual full rank factorization of $G(\lambda)$ is
\be\label{dual_frf} G(\lambda) = \widetilde X(\lambda) \widetilde R(\lambda) , \ee
where $\widetilde R(\lambda)$ is a full row rank \emph{coimage} of $G(\lambda)$ (i.e., $\widetilde R^T(\lambda)$ is the range of $G^T(\lambda)$) and $\widetilde X(\lambda)$ is a full column rank rational matrix. The dual factorization (\ref{dual_frf}) can be computed, by determining the full rank factorization of $G^T(\lambda)$, in the form  $G^T(\lambda) =  \widetilde R^T(\lambda) \widetilde X^T(\lambda)$.

\subsection{Minimum proper bases of the range space} The columns of the range matrix $R(\lambda)$ form a rational basis of $\mathcal{R}(G(\lambda))$. This basis is called \emph{minima}l if the McMillan degree of $R(\lambda)$ (i.e., the number of poles of $R(\lambda)$) is the least achievable one. To determine a minimal proper basis, we choose $\mathds{C}_g = \mathds{C} \cup \{\infty\}$ and
$\mathds{C}_b = \emptyset$, in which case, $R(\lambda)$ has no zeros. A stable minimal basis can be constructed in the form (\ref{range_stab}). A minimal inner basis, satisfying $R^\sim(\lambda)R(\lambda) = I_r$, can be computed for a suitable choice of $F$ and $W$ in (\ref{range_stab}) (see \cite{Oara00} for the continuous-time case, and \cite{Oara05} for the discrete-time case).\footnote{For a TFM $G(s)$ of a continuous-time system, the conjugate (or adjoint) is defined as $G^\sim(s) := G^T(-s)$, while  for the TFM $G(z)$ of a discrete-time system $G^\sim(z) := G^T(1/z)$.}

\subsection{Normalized coprime factorizations} A straightforward application of the minimal inner range computation is the determination of a \emph{normalized right coprime factorization} of an arbitrary $p\times m$ rational matrix $G(\lambda)$ as
\be\label{giofac:nrcf} G(\lambda) = N(\lambda)M^{-1}(\lambda), \ee
such that $N(\lambda)$ and $M(\lambda)$  are stable and $\left[\begin{smallmatrix} N(\lambda) \\ M(\lambda) \end{smallmatrix}\right]$ is inner (i.e., $N^\sim(\lambda)N(\lambda) + M^\sim(\lambda)M(\lambda) = I$).
The factors $N(\lambda)$ and $M(\lambda)$ can be computed from a minimal inner basis $R(\lambda)$ of the range of $\left[ \begin{smallmatrix} G(\lambda) \\ I_m \end{smallmatrix}\right]$ satisfying
\[ \ba{c} G(\lambda) \\ I_m \ea = R(\lambda)X(\lambda) , \]
with
\[ R(\lambda) = \ba{c} N(\lambda) \\ M(\lambda) \ea, \qquad X(\lambda) = M^{-1}(\lambda) . \]

\subsection{Moore-Penrose pseudo-inverse} Another straightforward application of inner minimal bases of range spaces is the computation of the  Moore-Penrose pseudo-inverse $G^\#(\lambda)$ of a rational matrix $G(\lambda)$. This computation  can be performed in three steps, using a simplified version of the approach described in \cite{Oara00}:
\begin{enumerate}
\item Compute a  full-rank factorization
\[ G(\lambda) =  U(\lambda)G_1(\lambda), \]
with $U(\lambda)$, a minimal inner range matrix, and $G_1(\lambda)$ full row rank.
\item Compute the dual full-rank factorization
\[ G_1(\lambda) = G_2(\lambda)V(\lambda), \]
with $V(\lambda)$, a minimal co-inner coimage (i.e., $V(\lambda)V^\sim(\lambda) = I$), and $G_2(\lambda)$ invertible.
\item Compute
\[ G^\#(\lambda) = V^\sim(\lambda) G_2^{-1}(\lambda) U^\sim(\lambda). \]
\end{enumerate}
The dual full-rank factorization at Step 2 can be simply determined by computing the full-rank factorization $G_1^T(\lambda) = V^T(\lambda)G_2^T(\lambda)$, with $V^T(\lambda)$, minimal inner range matrix.

\subsection{Inner-outer factorization}  Let $\mathds{C}_g$ be the appropriate ``stability'' domain representing the closed left half complex plane, including infinity, for a continuous-time system, or the closed unit disc centered in the origin, for a discrete-time system, and define $\mathds{C}_b$ as its complement $\mathds{C}_b = \mathds{C}\setminus\mathds{C}_g$.
The generalized inner--quasi-outer factorization of a $p\times m$ rational matrix $G(\lambda)$, with normal rank $r$, is a special full rank factorization
\be\label{iofac} G(\lambda) = G_i(\lambda)G_o(\lambda) ,\ee
where $G_i(\lambda)$ is a $p\times r$ stable \emph{inner }factor (i.e., $G_i^\sim(\lambda)G_i(\lambda) = I_r$) and $G_o(\lambda)$ is \emph{quasi-outer}, having full row rank and only zeros in $\mathds{C}_g$.

To compute the factorization (\ref{iofac}), we can determine the range matrix $R(\lambda)$ of $G(\lambda)$  in (\ref{range_stab}), by choosing $F$ and $W$ such that $R(\lambda)$ is inner (see \cite{Oara00} for the continuous-time case, and \cite{Oara05} for the discrete-time case). Once the inner factor $G_i(\lambda) := R(\lambda)$ is determined, the quasi-outer factor results as $G_o(\lambda) := X(\lambda)$, where $X(\lambda)$ has the form (\ref{full-row-rank}).

\section{Examples}
The described range computation approach has been  implemented as a MATLAB function \texttt{\bfseries grange}, which belongs to the free software collection of \emph{Descriptor Systems Tools} (\textbf{DSTOOLS}) \cite{Varg17a}. This function also allows to compute inner range bases, including minimal inner range bases. For the computation of the special Kronecker-like form (\ref{spec-klf}) of a system matrix pencil, the function \texttt{\bfseries gsklf} has been implemented, which allows several choices of $\mathds{C}_b$. For the computation of the involved Kronecker-like form, the function \texttt{\bfseries gklf} is available, which is based on the Algorithm 3.2.1 of \cite{Beel88}. This algorithm underlies the implementation available in the SLICOT library \cite{Benn99}, which served as basis for the mex-function \url{sl_klf}, which has been used to implement \texttt{\bfseries gklf}.

\subsection*{Example 1} This is Example 1 from \cite{Oara00} of the transfer function matrix of a continuous-time proper system:
\be\label{giofac:ex1} G(s) = {\def\arraystretch{2}
\left[\begin{array}{ccc} \displaystyle\frac{s - 1}{s + 2} & \displaystyle\frac{s}{s + 2} & \displaystyle\frac{1}{s + 2}\\ 0 & \displaystyle\frac{s - 2}{{\left(s + 1\right)}^2} & \displaystyle\frac{s - 2}{{\left(s + 1\right)}^2}\\ \displaystyle\frac{s - 1}{s + 2} & \displaystyle\frac{s^2 + 2\, s - 2}{\left(s + 1\right)\, \left(s + 2\right)} & \displaystyle\frac{2\, s - 1}{\left(s + 1\right)\, \left(s + 2\right)} \end{array}\right]} \, .
\ee
$G(s)$ has zeros at $\{1, 2, \infty\}$, poles at $\{-1, -1, -2, -2 \}$, and normal rank $r = 2$.

A minimum proper basis of the range of $G(s)$, computed with \texttt{\bfseries grange}, is
\[ R(s) = \frac{1}{s + 1.374}
\ba{rr}
1.552 s + 2.124 & 1.314 s + 1.817\\
0.5931 s + 1.186 & -0.758 s - 1.516 \\
2.145 s + 2.717 & 0.5558 s + 1.059
\ea , \]
has McMillan-degree 1 and no zeros. The full row rank factor $X(s)$, satisfying $G(s) = R(s)X(s)$, has McMillan degree 4, and zeros at $\{1, 2, -1.374, \infty\}$. The zero at  $-1.374$ is equal to the pole of $R(s)$.

If we include in the computed range $R(s)$, both finite unstable zeros of $G(s)$, then
$R(s)$ has precisely only these (unstable) zeros at $\{1,2\}$ and has McMillan degree 3, with poles at $\{ -1\pm 0.433\iu, -2\}$. If we determine an inner range $R(s)$, then the unstable zeros of $G(s)$ are reflected to symmetric positions in $\{-1,-2\}$ as zeros of $R(s)$ and the poles of $R(s)$ are at
$\{ -1, -1.732, -2 \}$. Since $R(s)$ is the inner factor of an inner--quasi-outer factorization of $G(s)$, it follows that the full row rank factor $X(s)$, satisfying $G(s) = R(s)X(s)$, is the quasi-outer factor. For reference purposes, we give the resulting inner factor
\[ R(s) = {\arraycolsep=1mm\def\arraystretch{2}
\ba{rr}
\displaystyle\frac{-0.6078 s^3 - 1.944 s^2 - 1.501 s + 0.6181}{(s+2) (s+1.732) (s+1)} &
\displaystyle\frac{0.5452 s^3 - 0.2354 s^2 - 2.743 s - 2.507}{(s+2) (s+1.732) (s+1)} \\
\displaystyle\frac{-0.1683 s^3 - 0.6903 s^2 + 0.673 s + 2.761}{(s+2) (s+1.732) (s+1)} &
\displaystyle\frac{-0.799 s^3 - 0.4361 s^2 + 3.196 s + 1.744}{(s+2) (s+1.732) (s+1)} \\
\displaystyle\frac{-0.7761 s^3 - 2.466 s^2 - 0.474 s + 1.999}{(s+2) (s+1.732) (s+1)}  &
\displaystyle\frac{-0.2538 s^3 + 0.1274 s^2 - 0.7092 s - 1.635}{(s+2) (s+1.732) (s+1)}
\ea } . \]

\subsection*{Example 2} This is Example 2 from \cite{Oara05} of the transfer function matrix of a discrete-time polynomial system:
\be\label{glsfg:ex1}  G(z) = {\left[\begin{array}{ccc} z^2 + z + 1 & 4\, z^2 + 3\, z + 2 & 2\, z^2 - 2\\ z & 4\, z - 1 & 2\, z - 2\\ z^2 & 4\, z^2 - z & 2\, z^2 - 2\, z \end{array}\right]
},
\ee
which has two infinite poles (i.e., McMillan-degree of $G(z)$ is equal to 2), a zero at 1, and has a minimal descriptor state-space realization of order 4.

A minimum proper basis of the range of $G(z)$, computed with \texttt{\bfseries grange}, is
\[ R(z) = \frac{1}{z + 0.3304}
\ba{cc}
-1.564 z - 0.8277 & 0.06277 z + 0.4338\\
-0.9414  &  1.25 \\
-0.9414 z & 1.25 z
\ea , \]
has McMillan-degree 1 and no zeros. The full row rank factor $X(z)$, satisfying $G(z) = R(z)X(z)$, has McMillan degree 2, and zeros at $\{-0.3304, 1\}$. Notice that the zero at  $-0.3304$ is equal to the pole of $R(z)$.

An inner range $R(z)$ results as
\[ R(z) =
{\def\arraystretch{2}\ba{rr}
-0.7614 & -0.6483\\
\displaystyle\frac{-0.4584}{z}  &  \displaystyle\frac{0.5384}{z} \\
-0.4584 & 0.5384
\ea }, \]
has McMillan-degree 1 and no zeros. The quasi-outer factor $X(z)$ results as
\[ X(z) = {\arraycolsep=1mm\ba{rrr}
-1.678 z^2 \!-\! 0.7614 z \!-\! 0.7614 & -6.713 z^2 \!-\! 1.367 z \!-\! 1.523 & -3.356 z^2\!+\!1.834 z\!+\! 1.523 \\
0.4285 z^2 \!-\! 0.6483 z \!-\! 0.6483 & 1.714 z^2 \!-\! 3.022 z \!-\! 1.297 & 0.8571 z^2 \!-\! 2.154 z\!+\! 1.297
\ea },\]
has McMillan degree 2 and zeros at $\{ 0, 1 \}$.

\section{Conclusions}
\label{sec:conclusions}

In this note we described a numerically reliable general approach to compute proper bases for the range space of a rational matrix  and, simultaneously, to produce a  complete full rank factorization of this matrix.  The underlying computational algorithms use descriptor system state-space realizations, for which, the only restriction is a certain stabilizability condition (always fulfilled when using irreducible realizations). The  techniques described in this note served for the implementation of robust computational software, which is part of \textbf{DSTOOLS}, a free collection of descriptor systems tools for MATLAB \cite{Varg17a}.

\newpage

\bibliographystyle{siamplain}

\begin{thebibliography}{10}

\bibitem{Beel88}
{\sc T.~Beelen and P.~{Van~Dooren}}, {\em An improved algorithm for the
  computation of {Kronecker's} canonical form of a singular pencil}, Linear
  Algebra Appl., 105 (1988), pp.~9--65.

\bibitem{Benn99}
{\sc P.~Benner, V.~Mehrmann, V.~Sima, S.~{Van Huffel}, and A.~Varga}, {\em
  {SLICOT} -- a subroutine library in systems and control theory}, in Applied
  and Computational Control, Signals and Circuits, B.~N. Datta, ed., vol.~1,
  Birkh\"auser, 1999, pp.~499--539.

\bibitem{Golu13}
{\sc G.~H. Golub and C.~F. {Van~Loan}}, {\em Matrix Computations, 4th Edition},
  John Hopkins University Press, Baltimore, 2013.

\bibitem{Kail80}
{\sc T.~Kailath}, {\em Linear Systems}, Prentice Hall, Englewood Cliffs, 1980.

\bibitem{Oara05}
{\sc C.~Oar\u{a}}, {\em Constructive solutions to spectral and inner-outer
  factorizations with respect to the disk}, Automatica, 41 (2005),
  pp.~1855--1866.

\bibitem{Oara97}
{\sc C.~Oar\u{a} and P.~V. Dooren}, {\em An improved algorithm for the
  computation of structural invariants of a system pencil and related geometric
  aspects}, Syst. Control Lett., 30 (1997), pp.~39--48.

\bibitem{Oara00}
{\sc C.~Oar\u{a} and A.~Varga}, {\em Computation of general inner-outer and
  spectral factorizations}, IEEE Trans. Automat. Control, 45 (2000),
  pp.~2307--2325.

\bibitem{Varg17a}
{\sc A.~Varga}, {\em {DSTOOLS -- The Descriptor System Tools for MATLAB}}.
\newblock
  \url{https://sites.google.com/site/andreasvargacontact/home/software/dstools}.

\bibitem{Varg17}
{\sc A.~Varga}, {\em Solving Fault Diagnosis Problems -- Linear Synthesis
  Techniques}, vol.~84 of Studies in Systems, Decision and Control, Springer
  International Publishing, 2017.

\bibitem{Verg81}
{\sc G.~Verghese, B.~L\'evy, and T.~Kailath}, {\em A generalized state-space
  for singular systems}, IEEE Trans. Automat. Control, 26 (1981), pp.~811--831.

\bibitem{Verg79}
{\sc G.~Verghese, P.~{Van~Dooren}, and T.~Kailath}, {\em Properties of the
  system matrix of a generalized state-space system}, Int. J. Control, 30
  (1979), pp.~235--243.

\end{thebibliography}

\end{document}